\documentclass[]{article}
\usepackage{anysize}
\marginsize{3cm}{3cm}{3cm}{3cm}
\usepackage{rotating}
\usepackage{amssymb}
\usepackage{amsmath}
\usepackage[Large]{subfigure}
\usepackage{graphicx}
\usepackage{lscape}
\newtheorem{theorem}{Theorem}

\newtheorem{definition}[theorem]{Definition}

\newtheorem{proposition}[theorem]{Proposition}
\newtheorem{remark}[theorem]{Remark}

\newenvironment{proof}[1][Proof]{\noindent\textbf{#1.} }{\ \rule{0.5em}{0.5em}}

\begin{document}


\title{Mixed Tempered Stable distribution}

\author{Edit Rroji\footnote{Email: e.rroji@campus.unimib.it} and Lorenzo Mercuri\footnote{Email: lorenzo.mercuri@unimi.it}} 

\maketitle

\begin{abstract}
In this paper we introduce a new parametric distribution, the Mixed Tempered Stable. It has the same structure of the Normal Variance Mean Mixtures but the normality assumption leaves place to a semi-heavy tailed distribution. We show that, by choosing appropriately the parameters of the distribution and under the concrete specification of the mixing random variable, it is possible to obtain some well-known distributions as special cases. 

We employ the Mixed Tempered Stable distribution which has many attractive features for modeling univariate returns. Our results suggest that it is enough flexible to accomodate different density shapes.
Furthermore, the analysis applied to statistical time series shows that our approach provides a better fit
than competing distributions that are common in the practice of
finance.

\end{abstract}

\section{Introduction}

The Stable distribution has gained a great popularity in modeling economic and  financial time series starting from the seminal work of \cite{Mandelbrot63}. However, empirical evidence is against using the normal distibution for daily data but cannot accomodate the heavy tailedness of the Stable distribution. From a mathematical point of view a drawback of the Stable distribution is that  only fractional moments of order $p \leq \alpha$  with $\alpha \in \left(0,2\right)$ exist and, consequently, the standard Central Limit Theorem does not hold. For this  reason, several researchers have considered the Tempered Stable distribution \cite{BookContTankov} as a valid alternative in modeling financial returns. \newline The Tempered Stable distribution can be obtained by multiplying the L\'evy density of an $\alpha$- Stable with a decreasing tempering function . The tail behavior changes from heavy to semi-heavy characterized by an exponential instead of a power decay. The existence of the conventional moments is ensured and it satisfies the conditions of the classical Central Limit Theorem. This is an advantage in asset return modeling since for monthly or annual data the normality assumption seems to be correct [see\cite{Cont2001} for survey of the stylized facts].

In this paper we propose a new distribution: the Mixed Tempered Stable (MixedTS henceforth).  The idea is to build a new distribution in a similar way of the Normal Variance Mean Mixture (NVMM see \cite{Barndorff}) where the normality assumption is substituted with the standardized Tempered Stable (see \cite{Kim2008}).  We show that the new distribution can overcome some limits of the NVMM. In particular, the asymmetry and the kurtosis do not depend only on the mixing random variable but also on the other variable that determines the distribution.\newline 
If the mixing random variable follows a Gamma distribution, the proposed model can have the Variance Gamma \cite{MS1990, Loregian2012}, the Tempered Stable \cite{BookContTankov} and the Geometric Stable \cite{Kozubowski97} as special cases for specific choice of the parameters. 

In order to understand how this distribution fits to real data, we consider two examples. In the first, we build a Garch(1,1) with MixedTS innovations and use it in modeling univariate  financial time series. In the second, we consider a multifactor model for the log returns of a fund which tries to replicate the performance of the S\&P 500 index. As factors, we consider the Global Industry Classification Standard (GICS) indexes developed by MSCI-Barra provider. We capture the GICS dependence structure using the Independent Component Analysis introduced by \cite{art:Comon1994} and developed by \cite{book:Hyvarinen2001}. \newline  
Once fixed the number of the Independent Components (ICs), we use the Mixed Tempered Stable to model each of them. The observed fund returns are the sum of the ICs and the single factor return density could be non-gaussian and/or semi-heavy tailed. 

The Independent Component Analysis has been already used in finance and in particular to model  interest rates term structures \cite{art:Bellini2003}. \cite{madan2006equilibrium} proposed a non-Gaussian factor model using ICA with components assumed to follow a Variance Gamma distribution. Our model can be seen as a generalization of the latter since the Variance Gamma is a particular case of the Mixed Tempered Stable distribution. Following the same approach as in \cite{book:Meucci2005}  and \cite{Rroji13}, the marginal contribution to risk of each factor can be easily computed for given homogeneous risk measures due to the independence assumption for the components.

The outline of the paper is as follows. In Section 2 we review the main characteristics of the Tempered Stable distribution and make some useful observations . In Section 3 we introduce the Mixed Tempered Stable distribution and its main features. Section 4 illustrates the fitting of the MixedTS distribution to financial time series and to the Independent Components. Section 5 concludes the paper.

\section{Tempered Stable distribution}

In this section we review the main features of the Tempered Stable distribution. A random variable $X$ follows a Tempered Stable distribution if its L\'{e}vy measure is given by:
\begin{equation}
\nu\left(dx\right)=\left(\frac{C_{+}e^{-\lambda_{+}x}}{x^{1+\alpha_{+}}}1_{x>0}+
\frac{C_{-}e^{-\lambda_{-}\left|x\right|}}{\left|x\right|^{1+\alpha_{-}}}1_{x<0}\right)dx
\label{temperedStableLmeasure}
\end{equation}
with $\alpha_{+}, \ \alpha_{-} \in \left(0,2\right)$ and $C_{+}, \ C_{-}, \ \lambda_{+}, \ \lambda_{-} \in \left(0,+\infty\right)$.
\newline The characteristic function is obtained by solving the integral  \cite{BookContTankov}:
\begin{eqnarray}
E\left[e^{iuX}\right]&=&\exp\left[iu\gamma+\int_{\Re}\left(e^{iux}-1-iux\right)\nu\left(dx\right)\right]\nonumber\\
&=&\exp\left\{iu\gamma+C_{+}\Gamma\left(-\alpha_{+}\right)\left[\left(\lambda_{+}-iu\right)^{\alpha_{+}}-\lambda^{\alpha_{+}}\right]\right.\nonumber\\
&+&\left. C_{-}\Gamma\left(-\alpha_{-}\right)\left[\left(\lambda_{-}+iu\right)^{\alpha_{-}}-\lambda^{\alpha_{-}}\right]\right\}\nonumber\\
\label{TempCharacteristicGeneral}
\end{eqnarray} 
where $\gamma \in \Re$.
As observed in \cite{KT2013}, for $\alpha_{+}, \ \alpha_{-} \ \in \left(0,1\right)$, 
the Tempered Stable is obtained as a difference of two independent one sided Tempered Stable distributions introduced in \cite{Twd1984}. The corresponding process has finite variation with infinite activity. The interest of researchers for this distribution is confirmed from the fact that many particular cases have been investigated in literature. For $C_{+}=C_{-}=C$ and $\alpha_{+}=\alpha_{-}=\alpha$, we get the CGMY distribution \cite{CGMY}. If we fix $\alpha_{+}=\alpha_{-}$ and $\lambda_{+}=\lambda_{-}$, we get the truncated L\'{e}vy flight introduced in \cite{Ko1995} while for $\alpha_{+}=\alpha_{-}\rightarrow 0^{+}$ we get the Bilateral Gamma distribution \cite{KT2008, KT2008A, KT2009}. Computing the limit for $\alpha_{+}=\alpha_{-}\rightarrow 0^{+}$, $C_{+}=C_{-}$ and $\lambda_{+}=\lambda_{-}$ we obtain the Variance Gamma distribution [see\cite{MS1990, Loregian2012} for estimation]. 
In this paper, we consider the same restrictions  as in \cite{Kim2008}, i.e : $\alpha_{+}=\alpha_{-}=\alpha$ and $\gamma=\mu-\Gamma\left(1-\alpha\right)\left(C_{+}\lambda^{\alpha-1}_{+}-C_{-}\lambda^{\alpha-1}_{-}\right)$
where the distribution of the r.v $X \sim CTS\left(\alpha, \ \lambda_{+}, \ \lambda_{-}, \ C_{+}, \ C_{-}, \ \mu \right)$ is called Classical Tempered Stable. For this r.v  $E\left(X\right)=\mu$ and its characteristic function is given by:
\begin{eqnarray*}
E\left[e^{iuX}\right]&=&\Phi\left(u; \ \alpha, \ \lambda_{+}, \ \lambda_{-}, \ C_{+}, \ C_{-}, \ \mu\right)\\
&=&\exp\left[iu\mu-iu\Gamma\left(1-\alpha\right)\left(C_{+}{\lambda_{+}}^{\alpha-1}-C_{-}{\lambda_{+}}^{\alpha-1}\right)\right.\\
&+&\left.C_{+}\Gamma\left(-\alpha\right)\left(\left(\lambda_{+}-iu\right)^{\alpha}-\lambda_{+}^{\alpha}\right)+C_{-}\Gamma\left(-\alpha\right)\left(\left(\lambda_{-}+iu\right)^{\alpha}-\lambda_{-}^{\alpha}\right)\right]
\end{eqnarray*}
The cumulant of order n for the r.v X can  be obtained by taking the derivative of the characteristic exponent:
\begin{equation}
c_{n}\left(X\right):= \left.\frac{1}{i^n} \frac{\partial^n}{\partial u^n} \ln \left(E\left[e^{iuX}\right]\right)\right|_{u=0}
\end{equation}
Given the parameters of the distribution and fixing $n=1$, we get:
\begin{equation}
c_{1}(X)=\mu
\end{equation}
and for $n\geq 2$:
\begin{equation}
c_{n}(X)=\Gamma (n-\alpha)(C_{+} \lambda_{+}^{\alpha-n}+(-1)^{n}C_{-} \lambda_{-}^{\alpha-n})
\label{eq:cumulant}
\end{equation}Using the cumulants we get the first four moments of the distribution:
\begin{equation}
\left\{
\begin{array}{l}
E\left(X\right)=c_{1}\left(X\right)=\mu\\
Var\left(X\right)=c_{2}\left(X\right)=\Gamma\left(2-\alpha\right)\left[C_{+}\lambda^{\alpha-n}_{+}+\left(-1\right)^{n}C_{-}\lambda^{\alpha-n}_{-}\right]\\
\gamma_{1}=\frac{c_{3}\left(X\right)}{c_{2}^{3/2}\left(X\right)}=\frac{\Gamma\left(3-\alpha\right)\left[C_{+}\lambda^{\alpha-3}_{+}-C_{-}\lambda^{\alpha-3}_{-}\right]}{c_{2}^{3/2}\left(X\right)} \\
\gamma_{2}=3+\frac{c_{4}\left(X\right)}{c_{2}^{2}\left(X\right)}=3+\frac{\Gamma\left(4-\alpha\right)\left[C_{+}\lambda^{\alpha-4}_{+}-C_{-}\lambda^{\alpha-4}_{-}\right]}{c_{2}^{2}\left(X\right)}\\
\end{array}
\right.
\end{equation}
From the skewness formula it can be noticed that the difference between $C_{+}\lambda_{+}^{\alpha-3}$ and $C_{-}\lambda_{-}^{\alpha-3}$ drives the asymmetry while for $C_{+}=C_{-}$ the sign of the skewness depends on the difference of the two tempering parameters $\lambda_{+}$ and $\lambda_{-}$.
\newline The following result shows the convergence of the Tempered Stable distribution to the symmetric $\alpha-$ Stable distribution. 

\begin{proposition}
\label{prop:geo}
For $\lambda_{+}=\lambda_{-}=\lambda$, $\mu=0$ and $C_{+}=C_{-}=C$, the Tempered Stable distribution converges to the symmetric Stable distribution when $\lambda$ goes to zero. 
\end{proposition}
For $\lambda_{+}=\lambda_{-}=\lambda$, $\mu=0$ and $C_{+}=C_{-}=C$, the characteristic exponent in \eqref{TempCharacteristicGeneral} becomes  the characteristic function of a symmetric Tempered Stable distribution:
\begin{equation}
E\left[\exp\left(iuX\right)\right]=\exp\left[iu\mu+C\Gamma\left(1-\alpha\right)\left[\left(\lambda-iu\right)^\alpha+\left(\lambda+iu\right)^\alpha-2\lambda^\alpha\right]\right].
\label{chsymTS}
\end{equation}
For $r=\sqrt{u^2+\lambda^2}$ and $\theta=arctg\left(\frac{u}{\lambda}\right)$:
\begin{eqnarray*}
E\left[\exp\left(iuX\right)\right]&=&\exp\left[C \Gamma\left(1-\alpha\right)\left[r^{\alpha}e^{-i\alpha\theta}+r^{\alpha}e^{i\alpha\theta}-2\lambda^\alpha\right]\right]\\ &=&\exp\left[C \Gamma\left(1-\alpha\right)\left[2r^{\alpha}\cos(\alpha\theta)-2\lambda^\alpha\right]\right]
\end{eqnarray*}where the last equality is due to the Euler equation  $\frac{e^{i\theta}+e^{-i\theta}}{2}=\cos\left(\alpha\theta\right)$. The limit for $\lambda\rightarrow 0^{+}$ gives a more compact formula for the characteristic function:
\begin{equation}
\lim_{\lambda\rightarrow 0^{+}}\exp\left[C\Gamma\left(1-\alpha\right)\left[2\left(u^2+\lambda^2\right)^{\frac{\alpha}{2}}\cos(\alpha\theta)-2\lambda^\alpha\right]\right]=\exp\left[C\Gamma\left(1-\alpha\right)\left[2\left|u\right|^{\alpha}\cos\left(\frac{\alpha\pi}{2}\right)\right]\right]
\end{equation}that in fact is the characteristic function of a symmetric Stable distribution.
\begin{remark}The r.v $X$ has zero mean and unit variance for $\mu=0$ and 
\begin{equation}
C=C_{+}=C_{-}=\frac{1}{\Gamma(2-\alpha) (\lambda_{+}^{\alpha-2}+\lambda_{-}^{\alpha-2})}
\label{stdCond}
\end{equation}The  distibution of the r.v  X is called standardized Classical Tempered Stable, i.e $X \sim stdCTS(\alpha,\lambda_{+},\lambda_{-})$.
\end{remark}It is useful to observe that in the standardized Classical Tempered Stable distribution $C$ is fully determined once given the values for  $\alpha$, $\lambda_{+}$ and $\lambda_{-}$. Its characteristic exponent, defined as $L_{stdCTS}\left(u;\alpha,\lambda_{+},\lambda_{-}\right)=log E[e^{iuX}]$, is:
\begin{equation}
L_{stdCTS}(u:\alpha,\lambda_{+},\lambda_{-})=\frac{{(\lambda_{+}-iu)}^{\alpha}-\lambda_{+}^{\alpha}+{(\lambda_{-}+iu)}^{\alpha}-\lambda_{-}^{\alpha}}{\alpha(\alpha-1)(\lambda_{+}^{\alpha-2}+\lambda_{-}^{\alpha-2})}+\frac{iu(\lambda_{+}^{\alpha-1}-\lambda_{-}^{\alpha-1})}{(\alpha-1)(\lambda_{+}^{\alpha-2}+\lambda_{-}^{\alpha-2})}
\label{exponentstd}
\end{equation}For $\alpha \rightarrow 2$, we get the characteristic exponent of the Normal distribution :
\begin{eqnarray*}
\lim_{\alpha \rightarrow 2}L_{stdCTS}(u:\alpha,\lambda_{+},\lambda_{-})=-\frac{u^2}{2}\\
\end{eqnarray*}
%
\begin{remark}Condition \eqref{stdCond} implies that the convergence to $\alpha-$ stable is not possible since the characteristic exponent would converge to zero.
\end{remark}If the tempering parameters depend on a strictly positive quantity h, the standardized Classical Tempered Stable distribution has the following property that makes it appealing for mixtures. 
\begin{proposition}
\label{propScalStdCTS}
Let $\tilde{X}\sim stdCTS\left(u:\alpha,\lambda_{+}\sqrt{h},\lambda_{-}\sqrt{h}\right)$ and $h \in \left(0,+\infty\right)$ 
then the random variable $Y\stackrel{d}{=}\sqrt{h}\tilde{X}$ has the following characteristic exponent: 
\begin{equation}
\ln{E\left[e^{iu Y}\right]}=h\left[\frac{\left(\lambda_{+}-iu\right)^{\alpha}-\lambda^{\alpha}_{+}+\left(\lambda_{-}+iu\right)^{\alpha}-\lambda^{\alpha}_{-}}{\alpha\left(\alpha-1\right)\left(\lambda_{+}^{\alpha-2}+\lambda_{-}^{\alpha-2}\right)}+\frac{iu\left(\lambda_{+}^{\alpha-1}-\lambda_{-}^{\alpha-1}\right)}{\left(\lambda_{+}^{\alpha-2}+\lambda_{-}^{\alpha-2}\right)}\right]
\label{charactExp}
\end{equation}
Moreover if $h \in \mathbb{N}$ we have:
\begin{equation}
Y\stackrel{d}{=}\sum_{j=1}^{h}X_{j}
\label{scaleprop}
\end{equation}
where $X_{j}$ are iid $stdCTS\left(\alpha, \ \lambda_{+}, \ \lambda_{-}\right)$
\end{proposition}
\begin{proof} The characteristic exponent \eqref{exponentstd} evaluated in $\sqrt{h}u$ gives:
\begin{equation*}
\ln E\left[e^{iuY}\right]=\left[\frac{\left(\lambda_{+}\sqrt{h}-iu\sqrt{h}\right)^{\alpha}-h^{\frac{\alpha}{2}}\lambda^{\alpha}_{+}+\left(\lambda_{-}\sqrt{h}+iu\sqrt{h}\right)^{\alpha}-h^{\frac{\alpha}{2}}\lambda^{\alpha}_{-}}{\alpha\left(\alpha-1\right)\left(h^{\frac{\alpha}{2}-1}\lambda_{+}^{\alpha-2}+h^{\frac{\alpha}{2}-1}\lambda_{-}^{\alpha-2}\right)}+\frac{i\sqrt{h}u\left(h^{\frac{\alpha-1}{2}}\lambda_{+}^{\alpha-1}-h^{\frac{\alpha-1}{2}}\lambda_{-}^{\alpha-1}\right)}{\left(h^{\frac{\alpha}{2}-1}\lambda_{+}^{\alpha-2}+h^{\frac{\alpha}{2}-1}\lambda_{-}^{\alpha-2}\right)}\right]
\end{equation*}
Factorize $h$:
\begin{equation*}
\ln E\left[e^{iuY}\right]=\left[\frac{h^{\frac{\alpha}{2}}\left(\lambda_{+}-iu\right)^{\alpha}-\lambda^{\alpha}_{+}+\left(\lambda_{-}+iu\right)^{\alpha}-\lambda^{\alpha}_{-}}{h^{\frac{\alpha}{2}-1}\alpha\left(\alpha-1\right)\left(\lambda_{+}^{\alpha-2}+\lambda_{-}^{\alpha-2}\right)}+\frac{iuh^{\frac{\alpha}{2}}\left(\lambda_{+}^{\alpha-1}-\lambda_{-}^{\alpha-1}\right)}{h^{\frac{\alpha}{2}-1}\left(\lambda_{+}^{\alpha-2}+\lambda_{-}^{\alpha-2}\right)}\right]
\end{equation*}
and through simple manipulations we obtain the result in \eqref{charactExp}.
\newline To prove the result in equation \eqref{scaleprop} we use the iid assumption for $X_{j}$, $j=1,2,...,h$ where $h \in \mathbb{N}$. The characteristic exponent of the random variable $\sum_{j=1}^{h}X_{j}$ becomes:
\begin{equation*}
\ln E\left[\exp\left(\sum_{j=1}^{h}X_{j}\right)\right]=h\ln E\left[\exp\left(X_{1}\right)\right]
\end{equation*}
where $X_{1} \sim stdCTS\left(\alpha, \ \lambda_{+}, \ \lambda_{-}\right)$. Using equation \eqref{exponentstd} we obtain the characteristic exponent of $\sqrt{h}\tilde{X}$ that implies the relation in equation \eqref{scaleprop}. 
\end{proof}

\section{Mixed Tempered Stable distribution}
In this section, using proposition \ref{propScalStdCTS}, we define a new distribution that is shown to have some nice mathematical and statistical features.
\begin{definition}
 We say that a continuous random variable Y follows a Mixed Tempered Stable distribution if:
\begin{equation}
Y\stackrel{\tiny{d}}{=}\sqrt{V}\tilde{X}
\label{Def:MixTempStab}
\end{equation}
where $\tilde{X}\left|V\right.\sim stdTS(\alpha,\lambda_{+}\sqrt{V},\lambda_{-}\sqrt{V})$. V is a L\'{e}vy distribution defined on the positive axis and its moment generating function (mgf) always exists. 
\end{definition}The logarithm  of the mgf is :
\begin{equation}
\Phi_{V}(u)=\ln\left[E\left[\exp\left(uV\right)\right]\right]
\label{Def:CharactSub}
\end{equation}
We apply the law of iterated expectation for the computation of the characteristic function for the new distribution :
\begin{eqnarray}
E\left[e^{iu\sqrt{V}\tilde{X}}\right]&=&E\left\{E\left[\left.e^{iu\sqrt{V}\tilde{X}}\right|V\right]\right\}\nonumber\\
&=&\exp\left[\Phi_{V}\left(L_{stdCTS}\left(u;\ \alpha, \ \lambda_{+}, \ \lambda_{-}\right)\right)\right]\nonumber\\
\label{FinalChar}
\end{eqnarray}
The characteristic function identifies the distribution at time one of a time changed L\'{e}vy process that as shown in \cite{Sato} and \cite{CarrWu} it is infinitely divisible. \newline
\begin{proposition}
The first four moments for the MixedTS are:
\begin{equation*}
\left\{ 
\begin{array}{l}
E\left[\sqrt{V}\tilde{X}\right]=0\\
Var\left[\sqrt{V}\tilde{X}\right]=E\left[V\right]\\
\gamma_{1}=\left(2-\alpha\right)\frac{\left(\lambda^{\alpha-3}_{+}-\lambda^{\alpha-3}_{-}\right)}
{\left(\lambda^{\alpha-2}_{+}+\lambda^{\alpha-2}_{-}\right)}E^{-1/2}\left[V\right]\\
\gamma_{2}=\left[3+\left(3-\alpha\right)\left(2-\alpha\right)\frac{\left(\lambda^{\alpha-4}_{+}+\lambda^{\alpha-4}_{-}\right)}
{\left(\lambda^{\alpha-2}_{+}+\lambda^{\alpha-2}_{-}\right)}\right]\frac{E\left[V^2\right]}{E^{2}\left[V\right]}\\
\end{array}
\right.
\end{equation*}
\end{proposition}In Figure \ref{fig:SignBehaviourSkewForVaryingLambda} we show the behaviour of the skewness for different combinations of $\lambda_{+}$ and $\lambda_{-}$ and fixed $\alpha$. The same is done in Figure \ref{fig:BehaviourKurtForVaryingLambda} for the kurtosis.

\bigskip
Insert here Figure \ref{fig:SignBehaviourSkewForVaryingLambda}.

\bigskip 
Insert here Figure \ref{fig:BehaviourKurtForVaryingLambda}. 

\bigskip

If we assume that $V \sim \Gamma(a,\sigma^2)$,  the characteristic exponent in \eqref{FinalChar} becomes:
\begin{eqnarray}
E\left[\exp\left(u\sqrt{V}X\right)\right]&=&\exp\left[-a\ln\left(1-\sigma^2\frac{{(\lambda_{+}-iu)}^{\alpha}-(\lambda_{+})^{\alpha}+{(\lambda_{-}+iu)}^{\alpha}-\left(\lambda_{-}\right)^{\alpha}}{\alpha \left(\alpha-1\right)\left((\lambda_{+})^{\alpha-2}+(\lambda_{-})^{\alpha-2}\right)}\right.\right.\nonumber\\
&-&\left.\left.\sigma^2\frac{iu(\lambda_{+}^{\alpha-1}-\lambda_{-}^{\alpha-1})}{(\alpha-1)(\lambda_{+}^{\alpha-2}+\lambda_{-}^{\alpha-2})}\right)\right]\nonumber\\
\end{eqnarray}
Using the scaling property for the Gamma distribution (see \cite{Johnson}), the r.v. $Y$ can be equivalently defined as:
\begin{equation}
Y\stackrel{\tiny{d}}{=}\sigma\sqrt{\hat{V}}\hat{X}
\label{def:alternative}
\end{equation} 
where $\hat{V} \sim \Gamma\left(a,1\right)$ and $\hat{X}\left|\hat{V}\right. \sim stdTS\left(\alpha, \sigma\lambda_{+}\sqrt{\hat{V}}, \sigma\lambda_{-}\sqrt{\hat{V}} \right)$.

The MixedTS with Gamma mixing density has as special cases some well known distributions widely applied in different fields. Indeed, for $\sigma=\frac{1}{\sqrt{a}}$ and computing the limit for $a$ going to infinity, we retrieve the standardized Classical Tempered Stable (see Figure \ref{fig:ConvergenzaCTS}). The symmetric Variance Gamma distribution is obtained by choosing $\alpha=2$ as shown in Figure \ref{fig:ConvergenzaVG}. By choosing:
\begin{eqnarray}
\lambda_{+}&=&\lambda_{-}=\lambda\nonumber\\
a&=&1\nonumber\\ 
\sigma&=&\lambda^{\frac{\alpha-2}{2}}\gamma^{\frac{\alpha}{2}}\sqrt{\left| \frac{\alpha\left(\alpha-1\right)}{\cos\left(\alpha\frac{\pi}{2}\right)}\right|}\nonumber\\
\label{conditionForGeoStable}
\end{eqnarray}and computing the limit for $\lambda \rightarrow 0^{+}$ we obtain the Geometric Stable distribution.
Substituting the conditions  \eqref{conditionForGeoStable} in the characteristic exponent definition, we have:
\begin{equation*}
E\left[\exp\left(iu\sqrt{V}X\right)\right]=\exp\left[-\ln\left(1-\frac{{(\lambda_{+}-iu)}^{\alpha}-(\lambda_{+})^{\alpha}+{(\lambda_{-}+iu)}^{\alpha}-\left(\lambda_{-}\right)^{\alpha}}{2\alpha \left(\alpha-1\right)}\right)\right]
\end{equation*}Applying the limit and following the same arguments of proposition \ref{prop:geo}, we get:
\begin{equation*}
E\left[\exp\left(iu\sqrt{V}X\right)\right]\rightarrow\left(1-\left|u\right|^{\alpha}\frac{cos\left(\alpha\frac{\pi}{2}\right)}{\alpha\left(\alpha-1\right)}\right)^{-1}
\end{equation*} for any $\alpha\neq1$. The convergence to the Geometric-Stable distribution is shown in Figure \ref{fig:GeovsMixedComp}.
\bigskip

Insert Figure \ref{ConvergenceMixedTempStable} here.
\bigskip

The comparison of the MixedTS with the asymmetric Variance Gamma distribution, requires  a more general version of  equation \ref{def:alternative}. In practice, we consider a new random variable $\chi$ as follows:
\begin{equation}
\chi \stackrel{\tiny{d}}{=} \mu_0+\mu \hat{V} + Y
\label{eq:GeneralNVMM}
\end{equation} where $Y$ is defined in equation \ref{def:alternative}. In this way, we can have the mean to be different from  zero and a model that is easier to be compared with the NVMM since the structure is the same.

\section{Empirical study}

In this section we empirically investigate the performance of the MixedTS distribution in modeling asset returns by considering two examples. In the first, we study the ability of the new model to capture the stylized facts observed in each asset return time series. We use the MixedTS for the innovations of a Garch(1,1)  and compare its performance with a Garch(1,1)  with Variance Gamma innovations.
\newline In the second, we build a multifactor model using the Independent Component Analysis and assume that each component follows a MixedTS distribution with Gamma mixing density. 

The dataset is composed by daily log returns of the Vanguard Fund Index which tries to replicate the performance of the S\&P 500. It seems quite natural to consider as portfolio risk factors the daily log-returns of the 10 GICS indexes since each of the S\&P500 members belongs to one of them. The data are daily log returns  ranging from 14-June-2010 to 20-September-2012 obtained from the Bloomberg data provider. 
 
In Table \ref{Statistical measures of log returns} we report the main statistics of the considered indexes. Looking to the empirical skewness and kurtosis, the departure from the normal hypothesis is evident.
\bigskip

Insert Table \ref{Statistical measures of log returns} here.

\bigskip
In the first example,  log returns $r_{i,t}$ are modeled using the classical Garch(1,1) as in \cite{Bollerslev1987} :
\begin{eqnarray}
r_{i,t} &=& \sigma_{i,t} \chi_{i,t}\\
\sigma_{i,t}^2 &=& \alpha_{0}+\alpha_{1}r_{i,t-1}^2+\beta_{1}\sigma^{2}_{i,t-1}\\
\end{eqnarray}
where $\chi_{i,t}$ follows the general MixedTS defined in equation \ref{eq:GeneralNVMM}. The model is compared with another Garch(1,1) with the same structure but $\chi_{i,t}$ is Variance Gamma distributed.
Using the quasi-maximum likelihood method \cite{Bollerslev1987}, implemented in the Garch Matlab toolbox, we get  the values for the parameters $\alpha_{0}, \ \alpha_{1},$ $\beta_{1}$ and  the volatility sequence $\sigma_{t}$  for $t=1,2,...,T$. We estimate the values for the MixedTS parameters on the residual sequence $\chi_{t}$ minimizing the mean squared error computed using the empirical and the theoretical densities. 
In the MixedTS case, the density is computed using the Fourier transform while for the Variance Gamma we apply the approximation proposed in \cite{Loregian2012}. As measures of fit we consider the Mortara index $A_{1}$, the quadratic  Pearson index $A_{2}$ and the root of mean squared errors $X_2$:
\begin{eqnarray*}
A_{1}&=&\frac{1}{n}\sum_{j=1}^{K}\left|n_{j}-\hat{n}_{j}\right|\\
A_{2}&=&\sqrt{\frac{1}{n}\sum_{j=1}^{K}\frac{\left(n_{j}-\hat{n}_{j}\right)^2}{\hat{n}_{j}}}\\
X_2 &=& \sqrt{\frac{1}{n}\sum_{j=1}^{K}\left(n_{j}-\hat{n}_{j}\right)^2}\\
\end{eqnarray*} 
where we consider $K$  classes for the $n$ observations, $n_{j}$ are the observed frequencies and $\hat{n}_{j}$ the theoretical frequencies for the classes.  
\newline The estimated parameters for the Vanguard Fund Index and GICS are reported in Table \ref{MixedTSVGGigs}.
\bigskip

Insert Table \ref{MixedTSVGGigs} here.

\bigskip
Values for $\alpha$  lower than two suggest that the empirical distribution cannot be modeled considering only the  Variance Gamma. The better fit of the MixedTS with respect to the Variance Gamma is confirmed by  comparing the fitting measures in Table \ref{MixedTSVGGigs}.

In the second example, we assume a linear relation of the Vanguard Fund returns with the ten risk factors built using daily log returns of the GICS indexes :
\begin{equation}
r_{p,t}=\sum_{i=1}^{10}\beta_{i}X_{i,t}+\epsilon_{t}
\end{equation}
where $r_{p,t}$ and $X_{i,t}$ contain respectively the log returns of the Vanguard Fund and of the i-th sector. In $\beta_{i}$ we put the exposure of the portfolio returns to the i-th factor while $\epsilon$ is the idiosyncratic noise term. Through a simple  Ordinary Least Squares regression, we get the exposures to the risk factors whose results are reported in Table \ref{CoefficientsAndCapitalization}. In particular, we observe that the $R^2$ is higher than $99\%$ suggesting that the considered risk factors explain almost all the variability present in the Vanguard Fund log returns. Moreover, looking at Table \ref{CoefficientsAndCapitalization}, we observe that the estimated factor exposures in $\beta$ are coherent with the market capitalization of each sector.
\bigskip

Insert Table \ref{CoefficientsAndCapitalization} here.

\bigskip
%
We apply the FastICA algorithm proposed by \cite{art:Hyvarinen99} to the GICS  and find the Independent Components (ICs) that  maximize the non-gaussianity condition present in the optimization algorithm. If we think to the ICs as the columns of a matrix $S$, the ten sectors time series can be seen as linear transformations of the independent signal sources. In the mixing matrix $A \in \Re^{n\times n}$ is contained the information about the weight of the single original source in the market sector, i.e $X=A S$. The methodology is closely related to the well known Principal Component Analysis (see \cite{Joliffe}). However, while in the latter we assume that the unknown factors are normally distributed, in the ICA analysis factors are identified by maximizing any measure of non-Gaussianity for each component. Portfolio return distribution is obtained as a linear combination of independent r.v's that can have different distributions. 
We use the Jarque-Bera test to check for non-normality even though the departure from normality is confirmed even by sample skewness and kurtosis reported in Table \ref{StatisticsIndComp}.
\bigskip

Insert here Table \ref{StatisticsIndComp}.

\bigskip 
In matrix notations, portfolio return $r_{p} \in \Re^{1\times t}$ can be decomposed in the form:
\begin{equation}
r_{p}=\beta^{\tiny{F}}F +\beta^{\tiny{N}}N + \epsilon.
\label{decomp}
\end{equation}
 with $F \in \Re^{l \times t}$ being the matrix containing the $l$ rows of the $S$ matrix containing the components we decide to be meaningful in the market and with $N \in \Re^{(n-l) \times t}$ the remaining components considered as noise. The new exposures $\beta^{\tiny{F}} \in \Re^{1 \times l}$ and $\beta^{\tiny{N} \in \Re^{ 1\times (n-l)}}$  are obtained by multiplying  of the initial exposures $\beta$ with the corresponding vectors in the mixing matrix $A$. \newline
%

The linear relation in the portfolio return model can be used to compute the marginal contribution to return/risk of each of the chosen IC (as done in \cite{book:Meucci2005} or in \cite{Rroji13}) or in a portfolio optimization problem as in \cite{madan2006equilibrium}. We emphasize the fact that our main focus is not on introducing a new method on how to use ICA in finance but to stress the flexibility of the Mixed Tempered Stable distribution that can allow to capture contemporaneously the different shapes of each IC. In Table \ref{MixMatr} we report the estimated mixing matrix  obtained using the FastICA algorithm and observe that the MixedTS can fit statistical time series having different shapes.

\bigskip

Insert here table \ref{MixMatr}.
\bigskip

The empirical densities of the independent components are shown in Figure \ref{fig:CompDist}. In Table \ref{Compnent} we give the MixedTS fitted parameters for each component and some measures of fit.  
\bigskip

Insert here Figure \ref{fig:CompDist}.
\bigskip

Insert here Table \ref{Compnent}.
\bigskip

The four ICs with the highest Jarque Bera statistic are considered as factors while the others are grouped in the noise term $\hat{\epsilon}= \beta^{\tiny{N}}N + \epsilon $. The independence of $\hat{\epsilon}$ and $F$ simplifies the computations since we are able to compute the characteristic function for the portfolio returns $r_{p}$: 
\begin{equation*}
E[e^{iur_{p}}]=E[e^{iu[\sum_{i=1}^l \beta_{i}^{F}F_{i}+\hat{\epsilon}}]]=\prod_{i=1}^{l}E[e^{iu\beta_{i}^{F}F_{i}}]E[e^{iu\hat{\epsilon}}]
\end{equation*}
 The Vanguard return density is reconstructed using the calibrated MixedTS distribution parameters for the factors and assuming normality for the noise. It is possible to recover the density of the r.v $r_{p}$ from its characteristic function using the Fourier Transform. For comparison, we plot the normal distribution fitted to the fund return density.\ref{fig:ConstructSign}
\bigskip 

Insert here Figure \ref{fig:ConstructSign}

\section{Conclusion}

In this article, we discussed various features of the Tempered Stable distribution and introduced a new parametric distribution named Mixed Tempered Stable. The existing Normal Variance Mean Mixture models are based on the normality assumption while we try to generalize this concept. In fact, we consider the Standardized Classical Tempered Stable  instead of the Gaussian distribution. The mixing r.v  is defined on the positive axis but we showed that if it is Gamma distributed the Variance Gamma, the Tempered Stable and the Geometric Stable distribution are special cases. Despite the fact that this distribution has nice features from a theoretical point of view, it allows a dependence of the standard higher moments not only on the mixing r.v but also on the Standardized Classical Tempered Stable distribution. As part of our investigation, we also perform a sensitivity analysis of the skewness and kurtosis on the model parameters.\newline
As a first step, in our empirical study, we model the univariate financial returns using a Garch(1,1) with MixedTS innovations and compare the results in terms of fitting with the same model but with Variance Gamma innovations.  
Finally, we investigate the fitting of the MixedTS distribution to  the time series of the statistical factors obtained by applying the FastICA algorithm on the ten GICS sectors. It  is important to have a flexible distribution for accomodating the differences in terms of asymmetry and tail heaviness in the ICs. The fitting measures involved in the analysis confirm the theoretical findings and justify the superior performance of our model.

\bibliographystyle{plain}

\begin{thebibliography}{10}

\bibitem{Barndorff}
O.E. Barndorff-Nielsen, J.~Kent, and M.~Sørensen.
\newblock Normal variance-mean mixtures and z-distributions.
\newblock {\em International Statistical Review}, 50:145--159, 1982.

\bibitem{art:Bellini2003}
F~Bellini and E~Salinelli.
\newblock Independent component analysis and immunization: an exploratory
  study.
\newblock {\em International Journal of theoretical and applied finance},
  6(7):721 -- 738, 2003.

\bibitem{Bollerslev1987}
T~Bollerslev.
\newblock A conditionally heteroskedastic time series model for speculative
  prices and rates of return.
\newblock {\em The Review of Economics and Statistics}, 69(3):542--547, 1987.

\bibitem{CGMY}
P.~Carr, H.~Geman, D.B. Madan, and M.~Yor.
\newblock The fine structure of asset returns: an empirical investigation.
\newblock {\em Journal of Business}, 75(2):305--332, 2002.

\bibitem{CarrWu}
P.~Carr and L.~Wu.
\newblock Time-changed levy processes and option pricing.
\newblock {\em Journal of Financial Economics}, 71(1):113--141, 2004.

\bibitem{art:Comon1994}
P.~Comon.
\newblock Independent component analysis, a new concept?
\newblock {\em Signal Process.}, 36(3):287--314, 1994.

\bibitem{Cont2001}
R.~Cont.
\newblock Empirical properties of asset returns: stylized facts and statistical
  issues.
\newblock {\em Quantitative Finance}, 1:223--236, 2001.

\bibitem{BookContTankov}
R.~Cont and P.~Tankov.
\newblock {\em {Financial Modelling with Jump Processes}}, volume~II.
\newblock Chapman \& Hall/CRC Financial Mathematics Series, New York, 2003.

\bibitem{art:Hyvarinen99}
A.~Hyv\"{a}rinen.
\newblock Fast and robust fixed-point algorithms for independent component
  analysis.
\newblock {\em IEEE Transactions on Neural Networks}, 10(3):626--634, 1999.

\bibitem{book:Hyvarinen2001}
A.~Hyv\"{a}rinen and E.~Oja.
\newblock {\em Independent Component Analysis}.
\newblock New York: Wiley, 2001.

\bibitem{Johnson}
N.L. Johnson, S.~Kotz, and N.~Balakrishnan.
\newblock {\em {Continuous Univariate Distributions}}, volume~I.
\newblock Wiley Series in Probability and Statistics, 2002.

\bibitem{Joliffe}
I.T. Jolliffe.
\newblock {\em {Principal Component Analysis}}, volume~II of {\em 68}.
\newblock Springer, second edition, 2002.

\bibitem{Kim2008}
Y.S. Kim, S.~T. Rachev, M.~L. Bianchi, and F.J. Fabozzi.
\newblock Financial market models with lévy processes and time-varying
  volatility.
\newblock {\em Journal of Banking and Finance}, 32(7):1363--1378, 2008.

\bibitem{Ko1995}
I.~Koponen.
\newblock Analytic approach to the problem of convergence of truncated l\'{e}vy
  flights towards the gaussian stochastic process.
\newblock {\em Physical Review E}, 52:1197–1199, 1995.

\bibitem{Kozubowski97}
T.~J. Kozubowski, K.~Podg\'{o}rski, and G.~Samorodnitsky.
\newblock Tails of l\'{e}vy measure of geometric stable random variables, 1997.

\bibitem{KT2008}
U.~K\"{u}chler and S.~Tappe.
\newblock Bilateral gamma distributions and processes in financial mathematics.
\newblock {\em Stochastic Processes and their Applications}, 118(2):261--283,
  2008.

\bibitem{KT2008A}
U.~K\"{u}chler and S.~Tappe.
\newblock On the shapes of bilateral gamma densities.
\newblock {\em Statistics and Probability Letters}, 78(15):2478--2484, 2008.

\bibitem{KT2009}
U.~K\"{u}chler and S.~Tappe.
\newblock Option pricing in bilateral gamma stock models.
\newblock {\em Statistics and Decisions}, 27:281--307, 2009.

\bibitem{KT2013}
U.~K\"{u}chler and S.~Tappe.
\newblock Tempered stable distributions and processes.
\newblock {\em Stochastic Processes and their Applications}, 2013.

\bibitem{Loregian2012}
A.~Loregian, L.~Mercuri, and E.~Rroji.
\newblock Approximation of the variance gamma model with a finite mixture of
  normals.
\newblock {\em Statistic \& Probability Letters}, 82(2):217--224, 2012.

\bibitem{madan2006equilibrium}
D.~B. Madan.
\newblock Equilibrium asset pricing: with non-gaussian factors and exponential
  utilities.
\newblock {\em Quantitative Finance}, 6(6):455--463, 2006.

\bibitem{MS1990}
D.B. Madan and B.~Seneta.
\newblock The vg model for share market returns.
\newblock {\em Journal of Business}, 63:511--524, 1990.

\bibitem{Mandelbrot63}
B.~B. Mandelbrot.
\newblock The variation of certain speculative prices.
\newblock {\em The Journal of Business}, 36, 1963.

\bibitem{book:Meucci2005}
A.~Meucci.
\newblock {\em Risk and Asset Allocation}.
\newblock Springer, Berlin, 2005.

\bibitem{Rroji13}
E.~Rroji.
\newblock {\em Risk attribution and semi-heavy tailed distributions}.
\newblock Phd thesis, University of Milano-Bicocca, 2013.

\bibitem{Sato}
K.~I. Sato.
\newblock {\em {L\'{e}vy Processes and Infinitely Divisible Distributions}},
  volume~II of {\em 68}.
\newblock Cambridge studies in advanced mathematics, 1998.

\bibitem{Twd1984}
M.~Tweedie.
\newblock An index whichg distinguishes between some important exponantial
  families.
\newblock In {\em Proc. Indian Statistical Institute Golden Jubilee
  International Conference,}, pages 579--604. J. Ghosh and J. Roy (Eds.), 1984.

\end{thebibliography}

\clearpage 
\begin{figure}[h]
	\centering
		\includegraphics[width=0.90\textwidth]{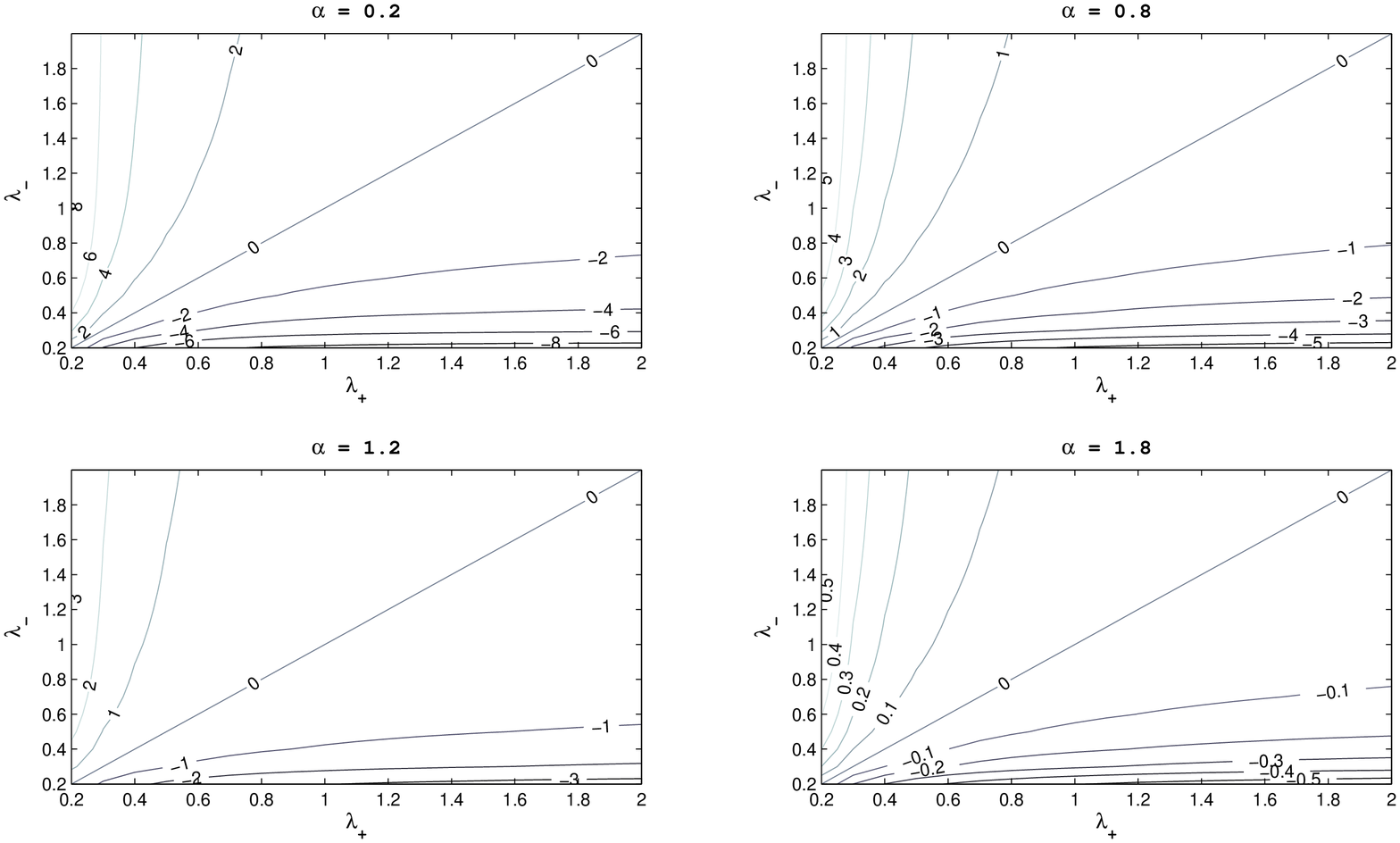}
\caption{\footnotesize{Consider the case when $ V \sim \Gamma(1,1)$ and fix some  values for $\alpha$. We plot the skewness curve level for different  combinations of $\lambda_{+}$ and $\lambda_{-}$ to have an idea of the possible skewness values. In the particular case when they coincide, the skewness is zero. The effect of an higher $\alpha$ is the reduction of the skewness level kept fixed values of the other parameters. The distribution of the MixedTS becomes symmetric for $\alpha=2$.}
\label{fig:SignBehaviourSkewForVaryingLambda}}
\end{figure}

\begin{figure}[h]
	\centering
		\includegraphics[width=0.90\textwidth]{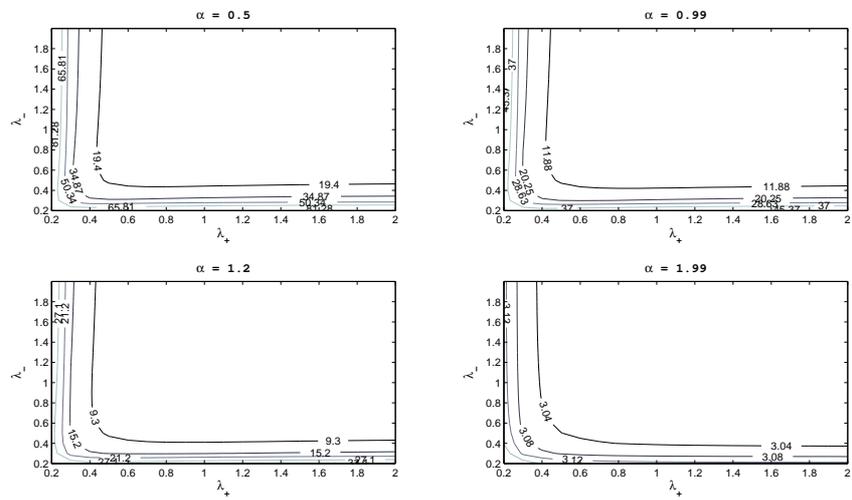}
\caption{\footnotesize{Consider the case when $ V \sim \Gamma(1,1)$ and fix some  values for $\alpha$. We plot the kurtosis values for different combinations of $\lambda_{+}$ and $\lambda_{-}$ to have an idea of the possible kurtosis values. The effect of an higher $\alpha$ is the reduction of the kurtosis level. If the fixed value for $\alpha$ is $1.9$ the curve level for kurtosis tend to be close to $3$ and the limiting case  of kurtosis equal to 3 is obtained for $\alpha=2$.} \label{fig:BehaviourKurtForVaryingLambda}}
\end{figure}
\clearpage

\begin{figure}[htbp]
\begin{center}
		\subfigure[\footnotesize{The symmetric VG distribution is a particular case of the MixedTS and it is obtained for $\alpha=2$. We fix $\mu_{0}=0$, $\mu=0$, $\sigma=1.2$, $a=1.7$, $\lambda_{+}=1.2$ and $\lambda_{-}=8$. In the figure we plot the MixedTS for different $\alpha$ values. The distribution associated is asymmetric but the limiting case not.}]{\includegraphics[width=0.9\textwidth]{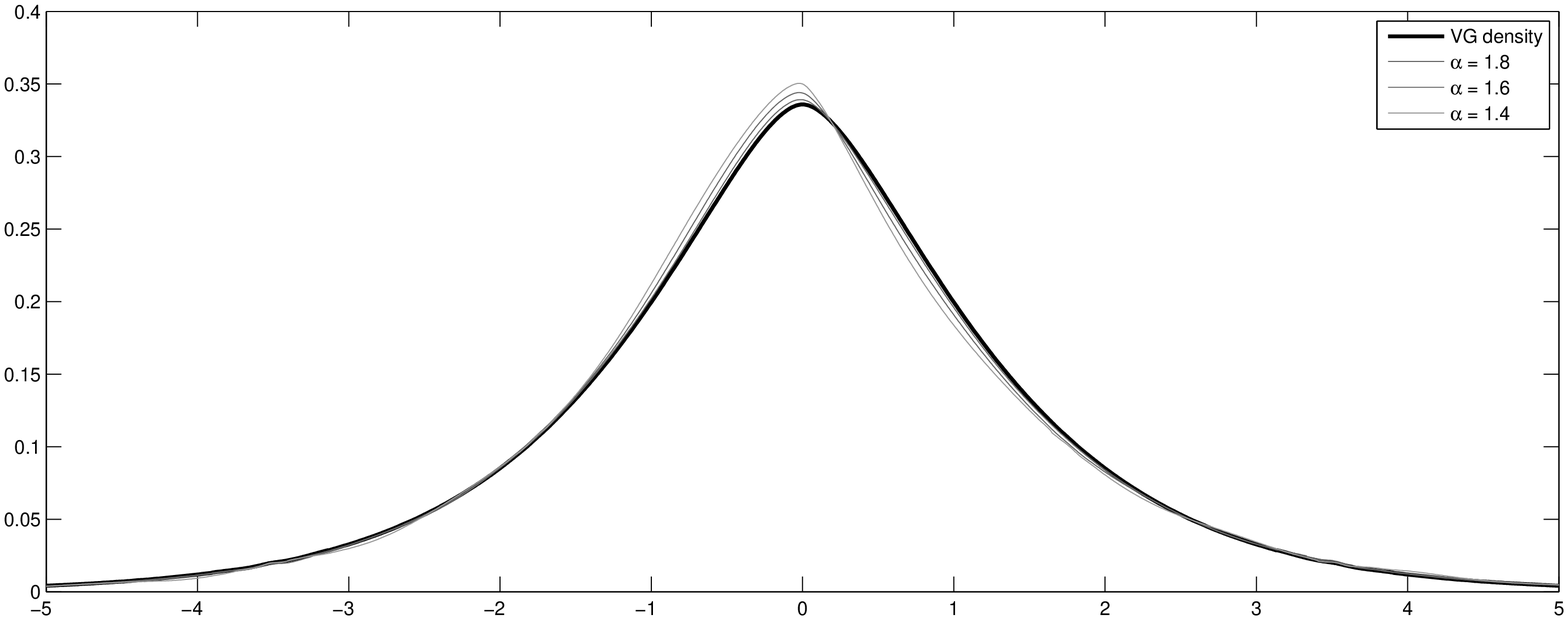}\label{fig:ConvergenzaVG}}
		\subfigure[\footnotesize{The symmetric CTS distribution is a particular case of the MixedTS and it is obtained for $\alpha=2$. We fix  $\mu_{0}=0$, $\mu=0$, $\lambda_{+}=1.2$ and $\lambda_{-}=8$ and $\alpha=1.4$. In the figure we plot the MixedTS for different values of $a$.}]{\includegraphics[width=0.9\textwidth]{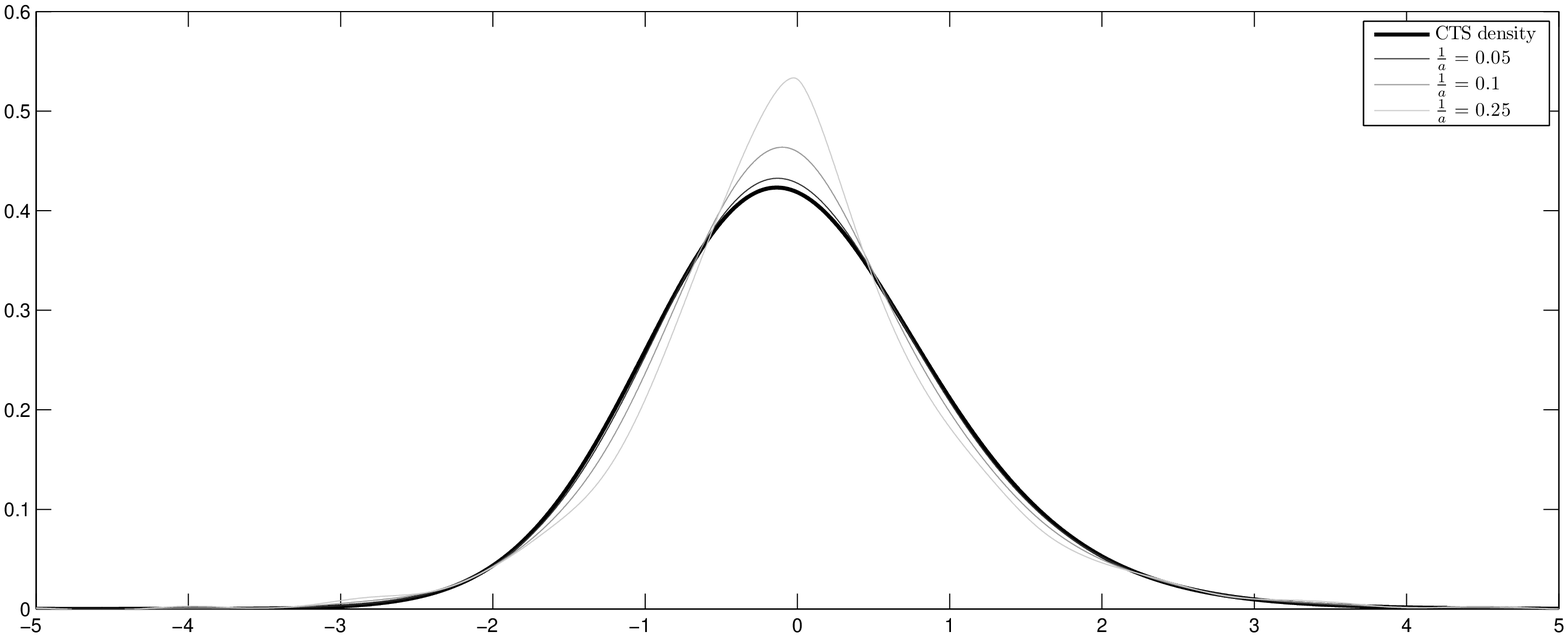}
	 \label{fig:ConvergenzaCTS}}
		\subfigure[\footnotesize{To show the convergence of the MixedTS distribution to the Geometric Stable one we plot it for different $\lambda$ values given $\mu_{0}=0$, $\mu=0$, $a=1$ and $\gamma=1$. The Geometric Stable distribution is obtained for $\sigma=\lambda^{\frac{\alpha-2}{2}}\gamma^{\frac{\alpha}{2}}\sqrt{\left| \frac{\alpha\left(\alpha-1\right)}{\cos\left(\alpha\frac{\pi}{2}\right)}\right|}$ and taking the limit for $\lambda \rightarrow 0$. Observe that as $\lambda$ gets smaller the tails get heavier. We cut the plot since in an open set around zero the Geometric Stable distribution has a peak going to $+ \infty$.}]{		\includegraphics[width=0.9\textwidth]{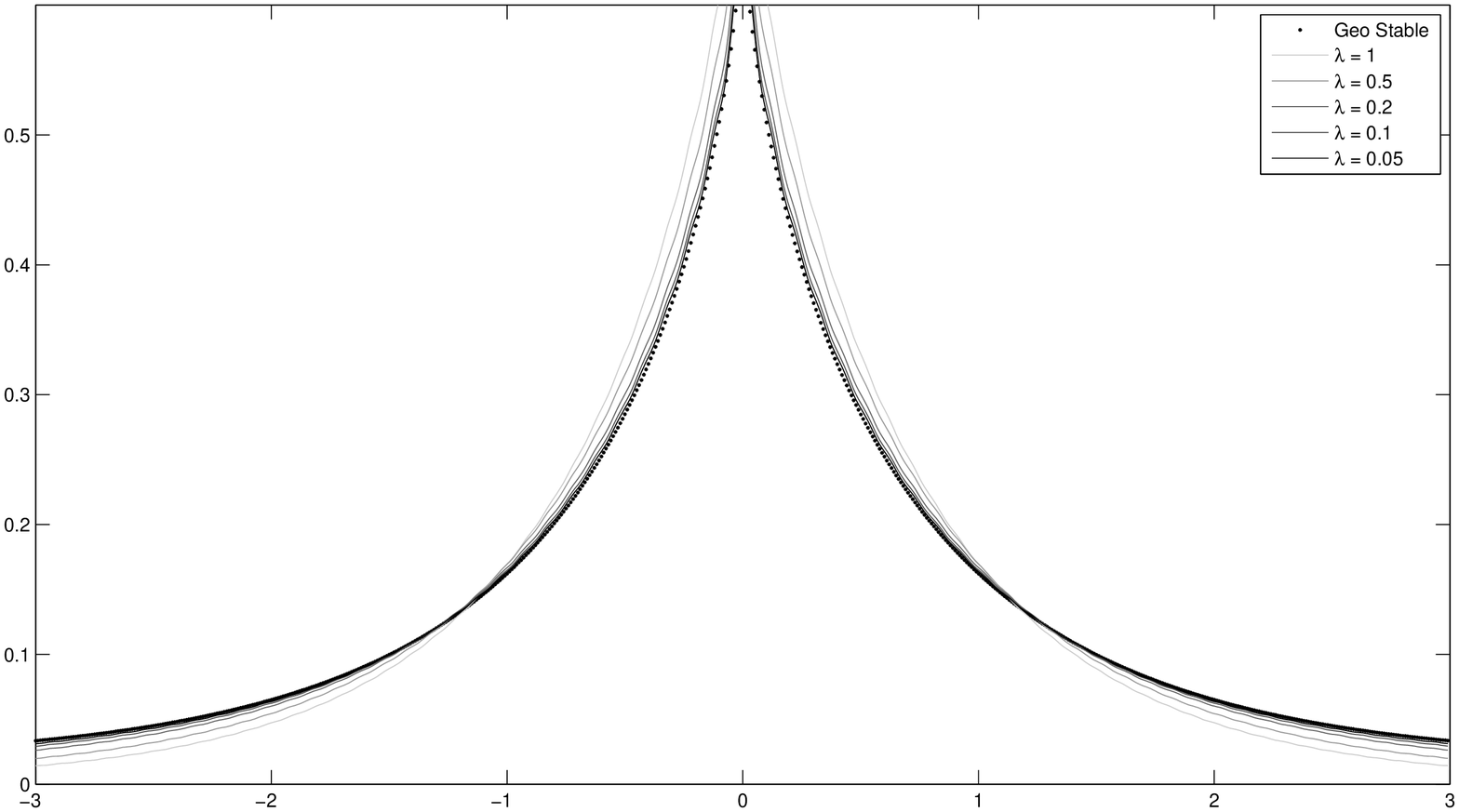}
	\label{fig:GeovsMixedComp}}
	\caption{\footnotesize{Special cases of the Mixed Tempered Stable distribution} \label{ConvergenceMixedTempStable}}
\end{center}
\end{figure}

\clearpage 

\begin{table}[htbp]
  \centering

    \begin{tabular}{lrrrrrr}
    \hline 
          & \multicolumn{1}{c}{Mean} & \multicolumn{1}{c}{Std} & \multicolumn{1}{c}{Skew} & \multicolumn{1}{c}{Ex-Kurt} & \multicolumn{1}{c}{Min} & \multicolumn{1}{c}{Max} \\
    \hline
    \multicolumn{1}{c}{\footnotesize{VFIAX}} & \footnotesize{0.0005} &\footnotesize{ 0.0121} &\footnotesize{ -0.4693} &\footnotesize{ 3.9683} & \footnotesize{-0.0688} &\footnotesize{ 0.0463} \\
    \multicolumn{1}{c}{\footnotesize{COND}} & \footnotesize{0.0007} & \footnotesize{0.0129} & \footnotesize{-0.5674} & \footnotesize{3.0566 }&\footnotesize{-0.0690} & \footnotesize{0.0472} \\
    \multicolumn{1}{c}{\footnotesize{CONS}} & \footnotesize{0.0006} & \footnotesize{0.0078} & \footnotesize{-0.3871} &\footnotesize{ 3.2413} & \footnotesize{-0.0390} & \footnotesize{0.0332} \\
    \multicolumn{1}{c}{\footnotesize{ENRS}} & \footnotesize{0.0006} & \footnotesize{0.0158} & \footnotesize{-0.3878} & \footnotesize{3.3909} & \footnotesize{-0.0864} & \footnotesize{0.0687} \\
    \multicolumn{1}{c}{\footnotesize{FINL}} &\footnotesize{ 0.0002} & \footnotesize{0.0175} & \footnotesize{-0.3415} &\footnotesize{ 4.1493 }& \footnotesize{-0.1052} & \footnotesize{0.0789} \\
    \multicolumn{1}{c}{\footnotesize{HLTH}} & \footnotesize{0.0006} & \footnotesize{0.0101} & \footnotesize{-0.4205} &\footnotesize{3.9130} & \footnotesize{-0.0540} & \footnotesize{0.0456} \\
    \multicolumn{1}{c}{\footnotesize{INDU}} & \footnotesize{0.0004} & \footnotesize{0.0142} & \footnotesize{-0.4390} & \footnotesize{2.7993} & \footnotesize{-0.0711} &\footnotesize{0.0495} \\
    \multicolumn{1}{c}{\footnotesize{INFT}} & \footnotesize{0.0007} & \footnotesize{0.0130} &\footnotesize{-0.2948} &\footnotesize{2.0849} & \footnotesize{-0.0596} & \footnotesize{0.0445} \\
    \multicolumn{1}{c}{\footnotesize{MATR}} & \footnotesize{0.0005} & \footnotesize{0.0159} & \footnotesize{-0.3575} & \footnotesize{2.5824} & \footnotesize{-0.0756} & \footnotesize{0.0593} \\
    \multicolumn{1}{c}{\footnotesize{TELS}} & \footnotesize{0.0007} & \footnotesize{0.0097} & \footnotesize{-0.2897} &\footnotesize{2.9581} & \footnotesize{-0.0550} & \footnotesize{0.0426}\\
    \multicolumn{1}{c}{\footnotesize{UTIL}} & \footnotesize{0.0004} & \footnotesize{0.0089} & \footnotesize{-0.1124} & \footnotesize{4.7863} & \footnotesize{-0.0563} & \footnotesize{0.0414} \\
    \hline
    \end{tabular}%
	\bigskip
	\caption{\footnotesize{The reported statistics for the fund VFIAX and for the GICS indexes show that the considered distributions are negatively skewed and the tails are heavier than those implied by the normal distibution.} \label{Statistical measures of log returns}}
\end{table}%

\begin{table}[htbp]
  \centering
    \begin{tabular}{rcccccccccc}
		\hline
          & \multicolumn{8}{c}{\footnotesize{Regression coefficients and Capitalization weights}}\\
    \hline
          & \multicolumn{1}{c}{\textbf{\footnotesize{COND}}} & \multicolumn{1}{c}{\textbf{\footnotesize{CONS}}} & \multicolumn{1}{c}{\textbf{\footnotesize{ENRS}}} & \multicolumn{1}{c}{\textbf{\footnotesize{FINL}}} & \multicolumn{1}{c}{\textbf{\footnotesize{HLTH}}} & \multicolumn{1}{c}{\textbf{\footnotesize{INDU}}} & \multicolumn{1}{c}{\textbf{\footnotesize{INFT}}} & \multicolumn{1}{c}{\textbf{\footnotesize{MATR}}} & \multicolumn{1}{c}{\textbf{\footnotesize{TELS}}} & \multicolumn{1}{c}{\textbf{\footnotesize{UTIL}}} \\
  
    $\footnotesize{\beta}$ & \footnotesize{0.1105} & \footnotesize{0.1154} & \footnotesize{0.1238} & \footnotesize{0.1442} & \footnotesize{0.1051} & \footnotesize{0.1145} & \footnotesize{0.1818} & \footnotesize{0.0378} & \footnotesize{0.0220} & \footnotesize{0.0415} \\
    \footnotesize{Cap weight (14/06/2010)} & \footnotesize{0.1103}      & \footnotesize{0.1165}      & \footnotesize{0.1206}      &  \footnotesize{0.1441}     &   \footnotesize{0.1190}    &   \footnotesize{0.1221}    &    \footnotesize{0.1800}   &   \footnotesize{0.0115}    &   \footnotesize{0.0213}    & \footnotesize{0.0546} \\
    \footnotesize{Cap weight (20/09/2012)} & \footnotesize{0.1108} & \footnotesize{0.1091}
 & \footnotesize{0.1127} & \footnotesize{0.1507} & \footnotesize{0.1228} & \footnotesize{0.099} & \footnotesize{0.1921} & \footnotesize{0.0349} & \footnotesize{0.0317} & \footnotesize{0.0362} \\
    \hline
    \end{tabular}%
		\caption{\footnotesize{We perform a regression analysis and obtain the factor exposures for our portfolio VFIAX where the factors are the sector indexes. The dataset is composed by closing prices ranging from 14/06/2010 to 21/09/2012. The model $R^2$ is $99.69\%$ meaning that the explanatory power of our factors is quite high. We report the capitalization weight of the factors at the begining and at the end of the study period. The factor exposures are in line with the average market capitalization for each sector. } \label{CoefficientsAndCapitalization}}  
\end{table}%


	\begin{sidewaystable}
	\centering
    \begin{tabular}{@{}lccccccccccc}
    \hline
         &  \multicolumn{1}{c}{\textbf{\footnotesize{VFIAX}}} & \multicolumn{1}{c}{\textbf{\footnotesize{COND}}} & \multicolumn{1}{c}{\textbf{\footnotesize{CONS}}} & \multicolumn{1}{c}{\textbf{\footnotesize{ENRS}}} & \multicolumn{1}{c}{\textbf{\footnotesize{FINL}}} & \multicolumn{1}{c}{\textbf{\footnotesize{HLTH}}} & \multicolumn{1}{c}{\textbf{\footnotesize{INDU}}} & \multicolumn{1}{c}{\textbf{\footnotesize{INFT}}} & \multicolumn{1}{c}{\textbf{\footnotesize{MATR}}} & \multicolumn{1}{c}{\textbf{\footnotesize{TELS}}} & \multicolumn{1}{c}{\textbf{\footnotesize{UTIL}}} \\
    \hline
    \footnotesize{$\mu_0$} & \footnotesize{-0.0681} & \footnotesize{-0.0318} & \footnotesize{0.1399} & \footnotesize{-0.3049} & 
		\footnotesize{-0.0830} & \footnotesize{-0.0338} & \footnotesize{-0.0655} & \footnotesize{0.0064} & \footnotesize{-0.0212} & 
		\footnotesize{0.6256} & \footnotesize{0.0936} \\	
    \footnotesize{$\mu$} & \footnotesize{0.0601} & \footnotesize{0.0227} & \footnotesize{-0.0454} & 
		\footnotesize{0.1310} & \footnotesize{0.0204} & \footnotesize{0.0780} & 
		\footnotesize{0.0232} & \footnotesize{-0.0311} & \footnotesize{0.0605} & \footnotesize{-0.1931}
		& \footnotesize{-0.0409} \\
    \footnotesize{$\sigma$} & \footnotesize{1.0530} & \footnotesize{0.7276} & \footnotesize{0.5038} 
		& \footnotesize{0.8314} & \footnotesize{0.7026} & \footnotesize{1.1109} &\footnotesize{0.7843} & \footnotesize{0.8554} & \footnotesize{1.0803} & \footnotesize{0.5487} & \footnotesize{0.5291} \\    
		\footnotesize{a}  & \footnotesize{1.1670} & \footnotesize{2.0313} & \footnotesize{3.8303} & \footnotesize{1.9440} & \footnotesize{2.2742} & \footnotesize{0.9718} & \footnotesize{1.8799} & \footnotesize{1.5514} & \footnotesize{1.2326} & \footnotesize{3.2875} & \footnotesize{3.4667} \\
		\footnotesize{$\lambda_{+}$} &   \footnotesize{1.0280} & \footnotesize{1.0384} & \footnotesize{1.0855} & \footnotesize{1.6044} & \footnotesize{1.0921} & \footnotesize{1.0000} & \footnotesize{1.0635} & \footnotesize{1.0540} & \footnotesize{0.9942} & \footnotesize{0.4083} & \footnotesize{0.9824}\\
		\footnotesize{$\lambda_{-}$} &\footnotesize{1.0311} & \footnotesize{1.0786} & \footnotesize{1.1733} & \footnotesize{0.4052} & \footnotesize{1.0961} & \footnotesize{1.0000} & \footnotesize{1.0801} & \footnotesize{1.0925} & \footnotesize{1.6001} & \footnotesize{1.9144} & \footnotesize{1.2202} \\
    \footnotesize{$\alpha$} & \footnotesize{1.4717} & \footnotesize{1.6663} & \footnotesize{1.9189} & \footnotesize{1.2897} & \footnotesize{1.7461} & \footnotesize{1.3000} & \footnotesize{1.6610} & \footnotesize{1.5913} & \footnotesize{1.3256} & \footnotesize{1.5053} & \footnotesize{1.8437} \\
     \footnotesize{A\_2\_MixedTS} &\footnotesize{0.0060} & \footnotesize{0.0055} & \footnotesize{0.0035} & \footnotesize{0.0065} & \footnotesize{0.0047} & \footnotesize{0.0055} & \footnotesize{0.0048} & \footnotesize{0.0046} & \footnotesize{0.0057} & \footnotesize{0.0061} & \footnotesize{0.0093} \\
    \footnotesize{X\_2\_MixedTS} & \footnotesize{0.0400} & \footnotesize{0.0333} & \footnotesize{0.0345} & \footnotesize{0.0413} & \footnotesize{0.0365} & \footnotesize{0.0317} & \footnotesize{0.0363} & \footnotesize{0.0370} & \footnotesize{0.0393} & \footnotesize{0.0363} & \footnotesize{0.0474}\\
    \footnotesize{A\_1\_MixedTS} & \footnotesize{0.0038} & \footnotesize{0.0036} & \footnotesize{0.0021} & \footnotesize{0.0042} & \footnotesize{0.0034} & \footnotesize{0.0037} & \footnotesize{0.0037} & \footnotesize{0.0038} & \footnotesize{0.0039} & \footnotesize{0.0034} & \footnotesize{0.0061} \\
    \footnotesize{A\_2\_VG} & \footnotesize{0.0062} & \footnotesize{0.0066} & \footnotesize{0.0066} & \footnotesize{0.0071} & \footnotesize{0.0057} & \footnotesize{0.0092} & \footnotesize{0.0058} & \footnotesize{0.0048} & \footnotesize{0.0063} & \footnotesize{0.0075} & \footnotesize{0.0069} \\
    \footnotesize{X\_2\_VG} & \footnotesize{0.0449} & \footnotesize{0.0346} & \footnotesize{0.0384} & \footnotesize{0.0435} & \footnotesize{0.0377} & \footnotesize{0.0345} & \footnotesize{0.0383} & \footnotesize{0.037} & \footnotesize{0.0385} & \footnotesize{0.0415} & \footnotesize{0.0555} \\
    \footnotesize{A\_1\_VG}  & \footnotesize{0.0042} & \footnotesize{0.0040} & \footnotesize{0.0041} & \footnotesize{0.0050}0 & \footnotesize{0.0042} & \footnotesize{0.0055} & \footnotesize{0.0044} & \footnotesize{0.0039} & \footnotesize{0.0042} & \footnotesize{0.0042} & \footnotesize{0.0047} \\
    \hline
    \end{tabular}%
		\bigskip
  \caption{\footnotesize{We fit the MixedTS and the VG distribution to the empirical density of each sector and obtain the corresponding parameters for both models.} \label{MixedTSVGGigs}}
\end{sidewaystable}

\clearpage 
\begin{table}[htbp]
  \centering
    \begin{tabular}{rccccccccccc}
    \hline
    \multicolumn{11}{c}{\footnotesize{Statistics ICs}} \\
    \hline
              \multicolumn{1}{c}{} & \footnotesize{\textbf{I}}     & \footnotesize{\textbf{II}}    & \footnotesize{\textbf{III}}   & \footnotesize{\textbf{IV}}    &\footnotesize{\textbf{V}}     & \footnotesize{\textbf{VI}}    & \footnotesize{\textbf{VII}}   & \footnotesize{\textbf{VIII}} & \footnotesize{\textbf{IX}}& \footnotesize{\textbf{X}} \\
    \multicolumn{1}{l}{\footnotesize{Skewness}} & \footnotesize{-0.6496} & \footnotesize{-0.1200} & \footnotesize{-0.5531} & \footnotesize{0.2913} & \footnotesize{-0.0349} & \footnotesize{-0.2916} & \footnotesize{0.0876} & \footnotesize{-0.1975} & \footnotesize{-0.0881} & \footnotesize{-0.0021} \\
    \multicolumn{1}{l}{\footnotesize{Kurtosis }} & \footnotesize{7.7030} & \footnotesize{7.5633} & \footnotesize{5.9752} & \footnotesize{5.9352} & \footnotesize{4.6628} & \footnotesize{4.1283} & \footnotesize{4.2410} & \footnotesize{3.7250} & \footnotesize{3.6370} & \footnotesize{3.4420} \\
    \multicolumn{1}{l}{\footnotesize{JB-Statistic}} & \footnotesize{546.5730} & \footnotesize{479.4040} & \footnotesize{231.3230} & \footnotesize{205.6020} & \footnotesize{63.5910} & \footnotesize{37.0450} & \footnotesize{36.0660} & \footnotesize{15.6540} & \footnotesize{10.0440} & \footnotesize{4.4830} \\
    \hline
    \end{tabular}%
  \bigskip
		\caption{\footnotesize{We report the skewness, kurtosis and Jarque Bera test statistic for each component.} \label{StatisticsIndComp}}  %
\end{table}%

\begin{table}[htbp]
  \centering
    \begin{tabular}{cccccccccc}
		\hline
          & \multicolumn{8}{c}{\footnotesize{Mixing Matrix}}\\
    \hline
   \textbf{\footnotesize{I}} & \textbf{\footnotesize{II}} & \textbf{\footnotesize{III}} & \textbf{\footnotesize{IV}} & \textbf{\footnotesize{V}} & \textbf{\footnotesize{VI}} & \textbf{\footnotesize{VII}} & \textbf{\footnotesize{VIII}} & \textbf{\footnotesize{IX}} & \textbf{\footnotesize{X}} \\
    \footnotesize{-0.0035}&\footnotesize{-0.0098} & \footnotesize{-0.0016} &\footnotesize{0.0024} &\footnotesize{-0.0018} &\footnotesize{0.0035} & \footnotesize{-0.0044} & \footnotesize{-0.0014} & \footnotesize{0.0028} & \footnotesize{0.0023} \\
    \footnotesize{-0.0008} & \footnotesize{-0.0059} &\footnotesize{-0.0001} &\footnotesize{0.0030} &\footnotesize{-0.0007} &\footnotesize{0.0029} & \footnotesize{0.0002} &\footnotesize{-0.0006} &\footnotesize{0.0021} &\footnotesize{-0.0014} \\
    \footnotesize{0.0030} &\footnotesize{-0.0126} &\footnotesize{-0.0006} & \footnotesize{0.0013} & \footnotesize{0.0000} & \footnotesize{0.0029} & \footnotesize{-0.0067} &\footnotesize{-0.0034} & \footnotesize{0.0034} & \footnotesize{-0.0015} \\
    \footnotesize{-0.0022} &\footnotesize{-0.0149} &\footnotesize{-0.0020} &\footnotesize{0.0047} &\footnotesize{0.0023} &\footnotesize{0.0030} &\footnotesize{-0.0052} & \footnotesize{-0.0027} & \footnotesize{-0.0027} & \footnotesize{0.0005} \\
    \footnotesize{-0.0021} &\footnotesize{-0.0083} & \footnotesize{-0.0017} & \footnotesize{0.0029} & \footnotesize{-0.0005} & \footnotesize{0.0002} &\footnotesize{-0.0026} & \footnotesize{0.0005} & \footnotesize{0.0027} &\footnotesize{-0.0020} \\
    \footnotesize{-0.0036} & \footnotesize{-0.0103} & \footnotesize{-0.0016} &\footnotesize{0.0029} & \footnotesize{0.0018} & \footnotesize{0.0035} & \footnotesize{-0.0055} & \footnotesize{-0.0037} & \footnotesize{0.0031} & \footnotesize{-0.0009} \\
    \footnotesize{-0.0028} & \footnotesize{-0.0089} & \footnotesize{-0.0030} & \footnotesize{0.0027} & \footnotesize{-0.0038} & \footnotesize{0.0019} & \footnotesize{-0.0046} & \footnotesize{-0.0050} & \footnotesize{0.0015} & \footnotesize{-0.0008} \\
    \footnotesize{-0.0019} & \footnotesize{-0.0113} & \footnotesize{-0.0023} & \footnotesize{0.0005} & \footnotesize{-0.0003} & \footnotesize{0.0060} & \footnotesize{-0.0082} & \footnotesize{-0.0018} & \footnotesize{0.0013} & \footnotesize{-0.0028} \\
    \footnotesize{-0.0029} & \footnotesize{-0.0077} & \footnotesize{0.0044} & \footnotesize{0.0013} & \footnotesize{-0.0010} & \footnotesize{0.0006} & \footnotesize{-0.0012} & \footnotesize{-0.0011} & \footnotesize{0.0007} & \footnotesize{-0.0017} \\
    \footnotesize{-0.0012} & \footnotesize{-0.0080} & \footnotesize{-0.0009} & \footnotesize{-0.0008} & \footnotesize{0.0004} & \footnotesize{0.0014} & \footnotesize{0.0018} & \footnotesize{-0.0011} & \footnotesize{0.0018} & \footnotesize{-0.0011} \\
    \hline

    \end{tabular}%
		\bigskip
		\caption{\footnotesize{Estimated mixing matrix obtained applying the FastICA algorithm to the dataset composed by the ten GICS indexes. Only 8 out of the 10 Independent Components are chosen to be used as factors. The remaining two are considered as noise.} \label{MixMatr}}  
\end{table}%

\bigskip

\begin{figure}[h]
	\centering
		\includegraphics[width=1.00\textwidth]{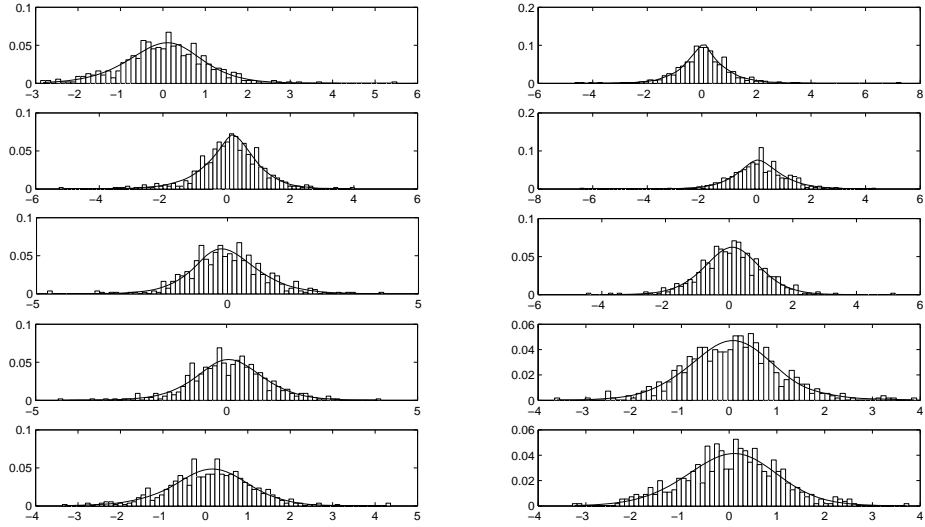}
	\caption{\footnotesize{The MixedTS is fitted to each IC empirical density. The fitted parameters are reported in Table  \ref{Compnent}. It is easy to observe that the MixedTS  can model return distributions which are asymmetric and/or fat tailed.} \label{fig:CompDist}}
\end{figure}

\bigskip

\begin{table}[htbp]
  \centering
    \begin{tabular}{rccccccccccc}
    \hline
          & \multicolumn{10}{c}{\footnotesize{MixedTS Parameters and Fitting Measures}} \\
    \hline
    \multicolumn{1}{c}{} & \footnotesize{I}     & \footnotesize{II}    & \footnotesize{III}   & \footnotesize{IV}    &\footnotesize{V}     & \footnotesize{VI}    & \footnotesize{VII}   & \footnotesize{VIII} & \footnotesize{IX}& \footnotesize{X} \\
    \multicolumn{1}{l}{\footnotesize{$\mu_0$}} & \footnotesize{-0.0771} & \footnotesize{0.1951} & \footnotesize{-0.1085} & \footnotesize{0.2336} & \footnotesize{0.5330} & \footnotesize{-0.6738} & \footnotesize{-0.4341} & \footnotesize{-1.5481} & \footnotesize{-5.1370} & \footnotesize{0.2181}\\
		\multicolumn{1}{l}{\footnotesize{$\mu$}} & \footnotesize{ 0.0275} & \footnotesize{-0.0915} & \footnotesize{0.0487} & \footnotesize{-0.1786} & \footnotesize{-0.3415} & \footnotesize{0.2133} & \footnotesize{0.1610} & \footnotesize{0.1025} & \footnotesize{0.1646} & \footnotesize{-0.0858}	\\
		\multicolumn{1}{l}{\footnotesize{$\sigma$}} & \footnotesize{1.0146} & \footnotesize{0.6666} & \footnotesize{0.7919} & \footnotesize{0.7814} & \footnotesize{0.8071} & \footnotesize{0.5983} & \footnotesize{0.6371} & \footnotesize{0.2595} & \footnotesize{0.1907} &\footnotesize{0.6358}	\\
    \multicolumn{1}{l}{\footnotesize{$a$}} & \footnotesize{1.2686} & \footnotesize{2.3377} & \footnotesize{2.1517} & \footnotesize{2.0305} & \footnotesize{1.4409} & \footnotesize{3.2953} & \footnotesize{2.8364} & \footnotesize{14.7360} & \footnotesize{30.4567} & \footnotesize{2.5446}\\
		\multicolumn{1}{l}{\footnotesize{$\lambda_+$}} & \footnotesize{1.0247} & \footnotesize{0.9146} & \footnotesize{1.8861} & \footnotesize{1.0256} & \footnotesize{1.0000} & \footnotesize{1.8062} & \footnotesize{2.2127} & \footnotesize{8.2816} & \footnotesize{9.9490} & \footnotesize{1.0890}\\
		\multicolumn{1}{l}{\footnotesize{$\lambda_-$}} & \footnotesize{0.9965} & \footnotesize{1.3665} & \footnotesize{0.1000} & \footnotesize{1.0431} & \footnotesize{1.0000} & \footnotesize{0.1000} & \footnotesize{0.6080} & \footnotesize{0.1000} & \footnotesize{0.1000} & \footnotesize{1.2074}\\
\multicolumn{1}{l}{\footnotesize{$\alpha$}} & \footnotesize{1.3724} & \footnotesize{1.7579} & \footnotesize{1.7150} & \footnotesize{1.5458} & \footnotesize{1.3000} & \footnotesize{1.6660} & \footnotesize{1.5057} & \footnotesize{1.7187} & \footnotesize{0.5000} & \footnotesize{1.8469} \\		
		\multicolumn{1}{l}{\footnotesize{$A_2$}} & \footnotesize{0.0082} & \footnotesize{0.0064} & \footnotesize{0.0059} & \footnotesize{0.0072} & \footnotesize{0.0076} & \footnotesize{0.0036} & \footnotesize{0.0077} & \footnotesize{0.0062} & \footnotesize{0.0052} & \footnotesize{0.0063}\\
			\multicolumn{1}{l}{\footnotesize{$X_2$}} & \footnotesize{0.0463} & \footnotesize{0.0781} & \footnotesize{0.0471} & \footnotesize{0.0512} & \footnotesize{0.0439} & \footnotesize{0.0314} & \footnotesize{0.0482} & \footnotesize{0.0434} & \footnotesize{0.0314} & \footnotesize{0.0412}\\
		\multicolumn{1}{l}{\footnotesize{$A_1$}} & \footnotesize{0.0048} & \footnotesize{0.0040} & \footnotesize{0.0038} & \footnotesize{0.0052} & \footnotesize{0.0045} & \footnotesize{0.0026} & \footnotesize{0.0047} & \footnotesize{0.0043} & \footnotesize{0.0037} & \footnotesize{0.0038}\\
    \hline
	\end{tabular}%
\bigskip
		\caption{\footnotesize{We fit the MixedTS distribution to the ICs considered as factors in our model. In the table we show the parameters and the fitting measures.} \label{Compnent}}  
\end{table}%

\begin{figure}[htbp]
	\centering
		\includegraphics[width=1.00\textwidth]{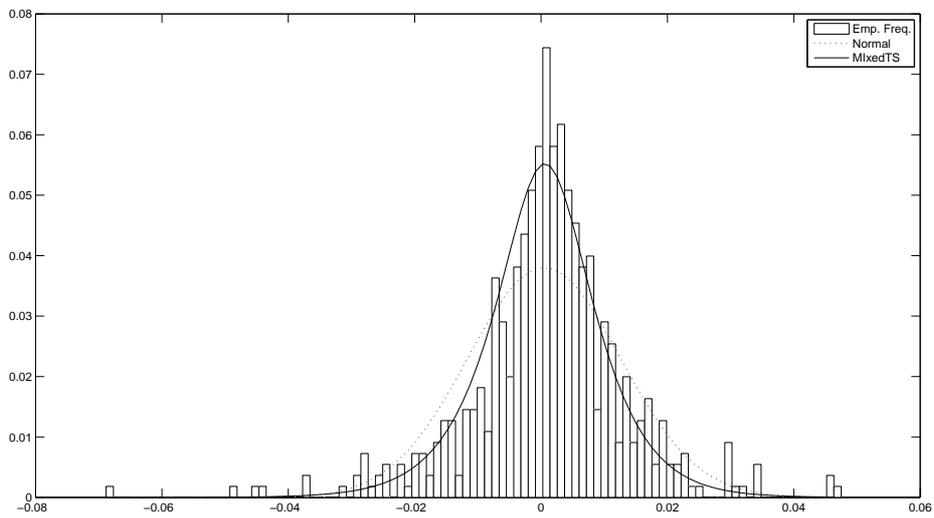}
		\caption{\footnotesize{The four ICs with the highest JB statistic are considered as relevant factors. The VFIAX return density is reconstructed using the MixedTS distribution for the factors and assuming normality for the noise. For comparison we plot the normal distribution fitted to the fund return density.} \label{fig:ConstructSign}} 
\end{figure}
\bigskip
\end{document}